\documentclass[11pt]{article}
\usepackage[a4paper,margin=30mm]{geometry}
\usepackage{amsmath,amssymb,amsthm,mathtools}
\usepackage[T1]{fontenc}
\usepackage{lmodern}
\usepackage{microtype}
\usepackage[hidelinks]{hyperref}
\usepackage{enumitem}
\usepackage{authblk}
\usepackage{graphicx} 
\usepackage{float} 
\usepackage{subcaption}
\usepackage{placeins}

\newtheorem{theorem}{Theorem}[section]
\newtheorem{proposition}[theorem]{Proposition}
\newtheorem{lemma}[theorem]{Lemma}
\newtheorem{corollary}[theorem]{Corollary}
\theoremstyle{definition}

\theoremstyle{remark}
\newtheorem{remark}[theorem]{Remark}

\title{Towards Digital Halftoning on Closed Manifolds--An Error Diffusion Scheme for the $2D$ Torus based on Sigma-Delta Quantization along the Rank-one Lattice}

\author[1,2]{Felix Krahmer}
\author[2]{Alessandro Lupoli}

\affil[1]{\small Department of Electrical Engineering and Information Technology, Technische Universit\"at Darmstadt}
\affil[2]{\small Department of Mathematics, Technical University of Munich}
\date{}

\begin{document}
\maketitle

\begin{abstract}
Digital halftoning aims to represent continuous-tone images by binary patterns while preserving their visually relevant low-frequency content. Among the many available approaches, error-diffusion methods implement noise shaping through causal feedback filters and can be interpreted as two-dimensional versions of the signal quantization paradigm Sigma--Delta modulation. On closed domains, however, the terminal state of the underlying recurrence relation need not match the initial one, producing boundary artifacts. We study this problem for bandlimited functions on the two-dimensional torus. By arranging all pixels along a single closed rank-one lattice, we replace the multiple mismatches associated with separately processed rows and columns with a single terminal contribution, while retaining exact reconstruction. For a uniform lattice with \(N=M^2+1\) points, we obtain first- and second-order error bounds of order \(N^{-1/2}\) and \(N^{-1}\). A suitable constant update eliminates the terminal mismatch and reduces the spatial localization of the error without changing these asymptotic orders. For fixed-direction rank-one lattices, the corrected first- and second-order reconstructions instead achieve rates \(N^{-1}\) and \(N^{-2}\). Numerical experiments illustrate a reduction in boundary artifacts compared with classical schemes applied on the Cartesian grid.
\end{abstract}

\section{Introduction}
\label{introduction}
Digital halftoning is the process of reproducing a continuous-tone image by means of
a binary pattern of black and white pixels. Since the output contains no intermediate
intensity values, gray levels cannot be represented pointwise; instead, they are
created by the local density and spatial arrangement of the dots. This mechanism
relies on the fact that the human visual system acts, to first approximation, as a
low-pass filter and therefore averages sufficiently fine spatial variations. As a
consequence, the relevant objective is not to make the representation error small at
every pixel, but to control its low-frequency component,  as this is sufficient for visual similarity.
As the coarseness of the binary representation prevents the total error from being small, this goal needs to be achieved by the so-called noise shaping paradigm:
the visually significant part of the
error is reduced by moving most of its energy towards higher spatial frequencies,
where it is attenuated by the visual system.

A number of algorithms explicitly build upon the idea of noise shaping. Specifically,  error-diffusion algorithms (see, for example, \cite{wang2009halftone,kite2000modeling,damera2001adaptive,krahmer2025mathematics,KrahmerVeselovska2023}) implement this principle through a recursive feedback mechanism. The
discrepancy produced at one pixel is propagated to subsequent pixels and influences
their binary values, so that nearby errors tend to cancel after spatial averaging.
This feedback structure is exactly analogous to the analog-to-digital (A/D) conversion scheme Sigma-Delta ($\Sigma\Delta$) modulation \cite{InoseYasuda1963}, which is why error diffusion has been interpreted in the literature as two-dimensional $\Sigma\Delta$
\cite{KiteEvansBovikSculley1997}. In this formulation, the sampled gray-scale values of the image take the role of the
input to the quantizer, while the halftoned image corresponds to the $1$-bit quantized representation, with each bit encoding whether the corresponding pixel is black or white.

That said, most of the algorithmic and theoretical developments of $\Sigma\Delta$ modulation have been achieved in the one-dimensional case. The main idea in $1D$ is that an ordered sequence of samples is processed recursively and an internal state variable records the accumulated quantization error. Beginning with the work of Inose and
Yasuda \cite{InoseYasuda1963}, this approach has become a standard method for
coarse analog-to-digital conversion, while its mathematical analysis for bandlimited
functions was initiated by Daubechies and DeVore
\cite{DaubechiesDeVore2003} and subsequently developed to obtain exponentially accurate one-bit schemes; see, for example,
\cite{Gunturk2003,DeiftGunturkKrahmer2011}. A notable exception to the focus on $1D$ is the work \cite{KrahmerVeselovska2023}, with its refinement [Sampta paper], which proves error guarantees for weighted combinations of $\Sigma\Delta$ schemes applied to solve the $2D$ problem. 

The underlying geometry in both $1D$ and $2D$ is, for all these papers, the standard Euclidean geometry, potentially restricted to a subset of $\mathbb{R}$ or $\mathbb{R}^2$. However, motivated by surface printing, the digital halftoning problem is also interesting for more general manifold domains. A challenge in this general framework arises when the manifold has a non-trivial topology and hence does not have a canonical identification with subsets of $\mathbb{R}$ or $\mathbb{R}^2$.

This issue is particularly relevant on a closed manifold, as there
is no distinguished boundary at which the recursion should naturally begin or end.
In $1D$, this problem has been studied in \cite{GrafKrahmerKrauseSolberg2023} with the focus on the simplest case where this issue arises, namely the circle. In order to apply a standard $\Sigma\Delta$ quantizer in this context, one needs to designate a starting point for the recurrence relation, which is then also the endpoint.
Consequently, one can encounter a mismatch of the state variable at this point, which translates into an additional local reconstruction error in the neighborhood of the endpoint. As shown in \cite{GrafKrahmerKrauseSolberg2023}, this is not just an artifact of the algorithm, but the error decay for a large oversampling rate is provably worse than in the Euclidean case.
This problem has also been investigated in the more general framework of finite frames; see \cite{wang2008sigma,bodmann2007frame,bodmann2007smooth,asano1999digital,benedetto2006second}. In this setting, eliminating the boundary contribution generally requires a non-translation-invariant dual: the reconstruction kernel has to be chosen depending on the sample location, which is not in line with the model commonly assumed for human vision.

As a remedy, \cite{GrafKrahmerKrauseSolberg2023,krause2018one} suggest not representing the original image but rather a perturbation of it obtained by adding a small constant. The motivation is that adding a small global constant is visually preferred over a local perturbation of the same magnitude. If the constant update is chosen carefully, the updated signal can be represented with provably better error decay.

In $2D$, a natural analog to the circle is the $2D$ torus. However, applying the $2D$ results of \cite{KrahmerVeselovska2023} requires not just a small number of initial conditions, but one needs to initialize the state variables at an entire vertical line and at an entire horizontal line, so the question of how to correct the boundary error is considerably more involved than in $1D$. 

In this paper, we address this issue by replacing the collection of row-wise orderings by a single trajectory containing all the sampling points. This allows us to then extend the program
initiated in \cite{GrafKrahmerKrauseSolberg2023}, where one-bit
$\Sigma\Delta$ quantization on the circle was considered as a first model for
error propagation on closed manifolds.

The idea of using a curve to impose a one-dimensional ordering on an image has already appeared in digital halftoning, notably through space-filling curves and related error-diffusion algorithms \cite{asano1999digital,velho1991digital}. In those works, the curve is primarily used as an algorithmic device for scanning the pixels and reducing directional artifacts. Here, instead, the trajectory is chosen so as to be compatible with both the topology of the torus and the Fourier reconstruction of the image.

Our trajectory of choice is a line on the torus, chosen such that the corresponding samples form a so-called rank-one lattice (see \cite{potts2021approximation,kuo2021function,gross2022sparse} for more details), which has recently played an important role in the construction of variants of the Fast Fourier Transform. 
More precisely, we sample the image on a rank-one lattice with $N=M^2+1$ points. As its nodes all lie on a line, they can be processed by one unchanged one-dimensional $\Sigma\Delta$ recursion. The particular choice $N=M^2+1$ also makes the lattice spatially uniform: its points are the centers of congruent rotated square pixels that partition the entire torus. At the same time, the lattice is reconstructing for the prescribed bandlimited space (see Section $2.1$ of \cite{potts2021approximation}), which leads to an exact equal-weight sampling formula. Thus, the same construction provides a closed ordering for quantization, a uniform pixelization of the torus, and an exact reconstruction operator.

Since all $N=M^2+1$ samples belong to a single trajectory, a boundary mismatch only occurs in a single point. In addition, this mismatch is small, as the size of this mismatch is inversely proportional to the sequence length. Even without any update, this favorable effect of the sequence length allows for first and second-order errors given by  
\[
O(M^{-1})=O(N^{-1/2})
\qquad\text{and}\qquad
O(M^{-2})=O(N^{-1}).
\]
After a constant correction, in the spirit of the $1D$ result, one can further improve this scaling.
However, this requires changing the discretization geometry and representing the image by rectangular pixels, rather than square ones. For that, we fix a line of large enough slope and increase the number of samples $N$ on that fixed trajectory. In this scenario, the reconstruction error scales as 
\[
O(N^{-1})
\qquad\text{and}\qquad
O(N^{-2}).
\]
for the first and second-order schemes, respectively. This gives a trade-off between spatially uniform pixels and faster asymptotic decay.


\section{Preliminaries}
\label{sec:preliminaries}

We begin by introducing the signal model and the main properties of
$\Sigma\Delta$ quantization that will be used throughout the paper. 
As a guiding example for a two-dimensional closed manifold domain, we consider the torus, so we study functions on the two-dimensional torus $\mathbb T^2$. 
Motivated by the fact that human vision can be modeled as a low-pass filter (as discussed in the introduction), we assume that the perceived images that we aim to approximate
are bandlimited on that domain.  This assumption
is directly aligned with A/D conversion theory, where it models the fact that different classes of signals occupy different frequency bands, and signal components outside the band of interest can be suppressed by a band-pass filter. To fully exploit this analogy, we represent the gray values by the symmetric interval $[-1,1]$, where $-1$ corresponds to black and $1$ to white.

Instead of $\mathbb T^2$, we work with $[0,2\pi)^2$ with opposite sides identified and regard functions on $\mathbb T^2$ as functions on $\mathbb R^2$ that are
$2\pi$-periodic in each variable.  For $f\in L^2(\mathbb T^2)$, we use the
normalization
\begin{equation}
\label{eq:fourier-coefficients-torus}
 \widehat f(k)
 =
 \frac{1}{(2\pi)^2}
 \int_{[0,2\pi)^2} f(t)e^{-ik\cdot t}\,dt,
 \qquad k\in\mathbb Z^2,
\end{equation}
where $k\cdot t=k_1t_1+k_2t_2$.  With this notation, the bandlimited images
considered in this paper belong to the following function space \begin{equation}
\label{eq:PWK-T2}
 P_K(\mathbb T^2)
 =
 \bigl\{ f\in L^2(\mathbb T^2):
 \operatorname{supp}(\widehat f)\subset[-K,K]^2\cap\mathbb Z^2
 \bigr\},
\end{equation}which can be interpreted as the periodic analog of the Paley-Wiener space

Equivalently, the signal model is the class of trigonometric polynomials with exponents contained in the square $[-K,K]^2$.  
It is well known that in this framework, an analog of Shannon's sampling theorem can be formulated using the Dirichlet kernel (Reference to book).
More precisely, let $N\ge 2K+1$, and let $y_{i,n}=f(t_{i,n})$ be the samples of
$f$ on the Cartesian grid
$t_{i,n}=(2\pi i/N,2\pi n/N)$, with $i,n=0,\ldots,N-1$.  Then $f$ can be
exactly recovered from these samples as
\begin{equation}
\label{eq:2d-sampling-prelim}
 f(t)
 =
 \frac{1}{N^2}
 \sum_{i,n=0}^{N-1}
 y_{i,n}\,\varphi_i^K(t_1)\varphi_n^K(t_2),
\end{equation}
where $\varphi_i^K(t)=\varphi^K(t-2\pi i/N)$, and $\varphi^K$ is the Dirichlet
kernel of bandwidth $K$,
\begin{equation}
\label{eq:phiK-prelim}
 \varphi^K(t)
 =
 \sum_{|m|\le K}e^{imt}
 =
 \frac{\sin\!\bigl((K+\tfrac12)t\bigr)}{\sin(t/2)}.
\end{equation}
Hence, before quantization, the finite array $(y_{i,n})$ fully represents the bandlimited function $f$.  This exact
sampling identity is the starting point for the quantization problem.

The goal is now to replace the gray values by black and white or, equivalently, the real-valued samples $y_{i,n}$ by quantized
values $q_{i,n}\in\{-1,1\}$ in such a way that $f$ is well approximated
by
\begin{equation}
\label{eq:quantized-reconstruction-prelim}
 f_q(t)
 =
 \frac{1}{N^2}
 \sum_{i,n=0}^{N-1}
 q_{i,n}\,\varphi_i^K(t_1)\varphi_n^K(t_2).
\end{equation}
That is, we aim to control the reconstruction error $e_q=f-f_q$, which can be expressed via \eqref{eq:quantized-reconstruction-prelim}
and \eqref{eq:2d-sampling-prelim} as
\begin{equation}
\label{eq:error-prelim}
 e_q(t)
 =
 \frac{1}{N^2}
 \sum_{i,n=0}^{N-1}
 (y_{i,n}-q_{i,n})\,\varphi_i^K(t_1)\varphi_n^K(t_2).
\end{equation}
Thus, the reconstruction error is obtained by applying the low-pass
reconstruction kernel to the quantization error $y-q$.  The aim is to choose
the quantized values so that $\|e_q\|_\infty$ becomes small as the
oversampling rate increases.

For such a coarse alphabet, just rounding the samples would not yield a meaningful representation; for example, all positive signals would be represented by the same array containing only the value $1$.  By contrast, $\Sigma\Delta$ quantization
chooses the quantized values in an interdependent way, in such a way that the quantization error is close to the null space of the low-pass reconstruction kernel. In other words, the error is shaped to be high-pass, which is why this paradigm is known as noise-shaping.

We briefly recall the one-dimensional framework underlying the analysis of
$\Sigma\Delta$ quantization, following
\cite{DaubechiesDeVore2003}. Let
$f\in PW_1(\mathbb R)$ and suppose that $f$ is sampled at rate
$\lambda\geq\lambda_0>1$. Then $f$ can be represented using a smoothed low-pass kernel given by a Schwartz function $\phi$ satisfying
\[
 \widehat{\phi}(\xi)=1
 \quad\text{for }|\xi|\leq\frac12,
 \qquad
 \widehat{\phi}(\xi)=0
 \quad\text{for }|\xi|\geq\frac{\lambda_0}{2},
\]
via the equality
\begin{equation}
\label{eq:one-dimensional-sampling}
 f(t)
 =
 \frac1\lambda
 \sum_{n\in\mathbb Z}
 y_n\phi\left(t-\frac{n}{\lambda}\right),
 \qquad
 y_n:=f\left(\frac{n}{\lambda}\right).
\end{equation}
Replacing the samples $y_n$ by one-bit values $q_n\in\{-1,1\}$ gives the
quantized reconstruction
\begin{equation}
\label{eq:one-dimensional-quantized-reconstruction}
 f_q(t)
 =
 \frac1\lambda
 \sum_{n\in\mathbb Z}
 q_n\phi\left(t-\frac{n}{\lambda}\right),
\end{equation}
and hence
\begin{equation}
\label{eq:one-dimensional-reconstruction-error}
 f(t)-f_q(t)
 =
 \frac1\lambda
 \sum_{n\in\mathbb Z}
 (y_n-q_n)\phi\left(t-\frac{n}{\lambda}\right).
\end{equation}
Thus, the reconstruction error is obtained by applying a low-pass filter to
the sample-wise quantization error $y-q$. The insight of  $\Sigma\Delta$ quantization is that if that error can be expressed as the finite difference of a bounded state variable, one can exploit the smoothness of the kernel to control the reconstruction error. More precisely, denoting by $\Delta$ the backward difference operator,
\[
 (\Delta u)_n=u_n-u_{n-1},
\]
an $r$-th order scheme is driven by a state sequence $u$ defined via the equality
\begin{equation}
\label{eq:standard-noise-shaping}
 y-q=\Delta^r u.
\end{equation}
Substitution into \eqref{eq:one-dimensional-reconstruction-error}, followed by
discrete summation by parts, transfers the differences from the state sequence
to the reconstruction kernel and yields the standard estimate
\begin{equation}
\label{eq:one-dimensional-sigma-delta-error}
 \|f-f_q\|_{L^\infty(\mathbb R)}
 \leq
 \lambda^{-r}
 \|\phi^{(r)}\|_{L^1(\mathbb R)}
 \|u\|_{\ell^\infty};
\end{equation}
see \cite{DaubechiesDeVore2003}. In particular, a
uniform bound on the state sequence $u$ leads to $r$-th order decay of the
reconstruction error as the oversampling rate increases.

A key ingredient for our method is the class of generalized $\Sigma\Delta$ schemes described by a causal feedback
filters $h=(h_j)_{j\geq0}\in\ell^1(\mathbb N_0)$ with $h_0=0$. Given an ordered
input sequence $y$, the quantization rule is
\begin{equation}
\label{eq:generalized-sigma-delta}
 \begin{aligned}
  q_n
  &=
  \operatorname{sign}\bigl((h*v)_n+y_n\bigr),\\
  v_n
  &=
  (h*v)_n+y_n-q_n,
 \end{aligned}
 \qquad
 (h*v)_n=\sum_{j\geq1}h_jv_{n-j},
\end{equation}
where $\operatorname{sign}(0)=1$ and $v$ denotes the internal state updated by
the recursion. Following \cite{Gunturk2003}, such a scheme yields a similar $r$-th order error decay if there exists $g\in\ell^1(\mathbb N_0)$ such that
\begin{equation}
\label{eq:generalized-order-condition}
 \delta_0-h=\Delta^r g.
\end{equation}
Indeed, the recursion gives
\[
 y-q=(\delta_0-h)*v,
\]
and therefore the reparametrized state $u:=g*v$ satisfies
\[
 y-q=\Delta^r u,
 \qquad
 \|u\|_{\ell^\infty}
 \leq
 \|g\|_{\ell^1}\|v\|_{\ell^\infty}.
\]
Combining this relation with
\eqref{eq:one-dimensional-sigma-delta-error} gives
\begin{equation}
\label{eq:generalized-reconstruction-error}
 \|f-f_q\|_{L^\infty(\mathbb R)}
 \leq
 \lambda^{-r}
 \|\phi^{(r)}\|_{L^1(\mathbb R)}
 \|g\|_{\ell^1}
 \|v\|_{\ell^\infty}.
\end{equation}
The reconstruction accuracy is therefore governed by two quantities: the size of the implemented state \(v\) and the norm \(\|g\|_{\ell^1}\). A scheme is said to be stable if \(v\) remains uniformly bounded, independently of the length of the input sequence. This property is essential for obtaining reconstruction guarantees, since \(\|v\|_{\ell^\infty}\) enters directly into the corresponding error estimate.

A standard sufficient condition for stability is
\begin{equation}
\label{eq:gunturk-stability-criterion}
\|h\|_{\ell^1}+\|y\|_{\ell^\infty}\leq 2.
\end{equation}
If the initial states lie in \([-1,1]\), this condition guarantees that
\(\|v\|_{\ell^\infty}\leq1\); see \cite{Gunturk2003}. 
It has been demonstrated empirically that optimizing the filters in view of conditions (14) and (15), gives rise to sparse filters $h$. This has motivated the study of filters of minimal support, namely, filters with the smallest number of nonzero coefficients compatible with the prescribed order; see \cite{Gunturk2003,deift2011optimal}. 
In the first-order case, the only such filter is given by \(h=(1)\), for the second order, these filters are given by 
\begin{equation}\label{second_order_minimal}
    h^{(k)}
 =
 \left(\frac{k+1}{k},0,\ldots,0,-\frac1k\right),
 \qquad k=1,2,\ldots,
\end{equation}

where \(h^{(k)}\) has length \(k+1\) and only two nonzero coefficients. Since

$$
 \|h^{(k)}\|_{\ell^1}=1+\frac{2}{k},
$$

the sufficient stability condition above is satisfied whenever

$$
 k\geq\frac{2}{1-\|y\|_{\ell^{\infty}}}.
$$

The choice of filter must therefore balance the stability requirement with the control of \(\|g\|_{\ell^1}\), which determines the constant in the reconstruction estimate. 

Minimally supported filters have been the key ingredient of the weighted $\Sigma\Delta$ schemes for $2$-dimensional signals proposed in \cite{KrahmerVeselovska2023}, and they will also play a key role in this paper.

\section{Boundary effects of Sigma-Delta schemes on manifolds}
\label{sec:rowwise-error}
In the $2$-dimensional torus, fixing one coordinate of a bivariate bandlimited function produces a one-dimensional bandlimited function of the remaining variable. Each row or column of the Cartesian sampling grid can therefore be regarded as the sampling of a bandlimited signal on the circle. Coarse quantization in such a scenario has been studied in \cite{GrafKrahmerKrauseSolberg2023}.

That work observes that applying $\Sigma\Delta$ in this context gives rise to boundary effects near the point where the scheme is initiated. These effects are not specific to $\Sigma\Delta$, but in general, unavoidable: the paper shows
that there exist bandlimited signals---in fact, even constant signals---for which the periodic reconstruction error is bounded from below by a multiple of \(N^{-1}\), independently of how the binary values are assigned to the \(N\) samples. The obstruction is therefore not tied to a particular quantization rule. It reflects the fact that a finite binary sequence does not, in general, close consistently when interpreted as the quantization of a signal on the circle.

This mechanism is particularly relevant for second-order schemes. Away from the endpoints, second-order noise shaping produces the expected quadratic cancellation. Terminal contributions, however, may remain on the \(N^{-1}\) scale and can therefore dominate the interior \(N^{-2}\) contribution. 
This problem can be avoided when the signal is altered by adding a suitable constant (which, arguably, is less severe than a non-constant error of comparable size), see \cite{GrafKrahmerKrauseSolberg2023}. Under an appropriate stability assumption, this modification yields second-order accuracy with respect to the shifted signal.

In $2D$, however, the suitable alterations can differ from line to line and column to column, and it is not clear how a global benign alteration can be designed.
Without such an alteration, one can still phase boundary effects, but this time on $1D$ boundary, rather than just a point. In Figure \ref{image_boundary1} we illustrate this for the weighted \(\Sigma\Delta\) schemes introduced in
\cite{KrahmerVeselovska2023}, a class of $\Sigma\Delta$ schemes specifically designed for halftoning, which combine one-dimensional noise-shaping relations along several directions of a Cartesian grid. More precisely, second-order weighted $\Sigma\Delta$ schemes use a filter that is a convex combination of horizontal, vertical, and diagonal second-order filters, so their basic building blocks remain finite one-dimensional recursions.

We observe that although the input functions and the reconstruction kernel are periodic, the reconstruction error develops pronounced ridges near the seams at which the finite directional recursions are initialized or interrupted. Our interpretation is that different rows, columns, and diagonals each carry their own terminal mismatch, and in the final two-dimensional reconstruction, the contributions from these trajectories are superimposed. 

\begin{figure}
\centering
\includegraphics[width=0.9\linewidth]{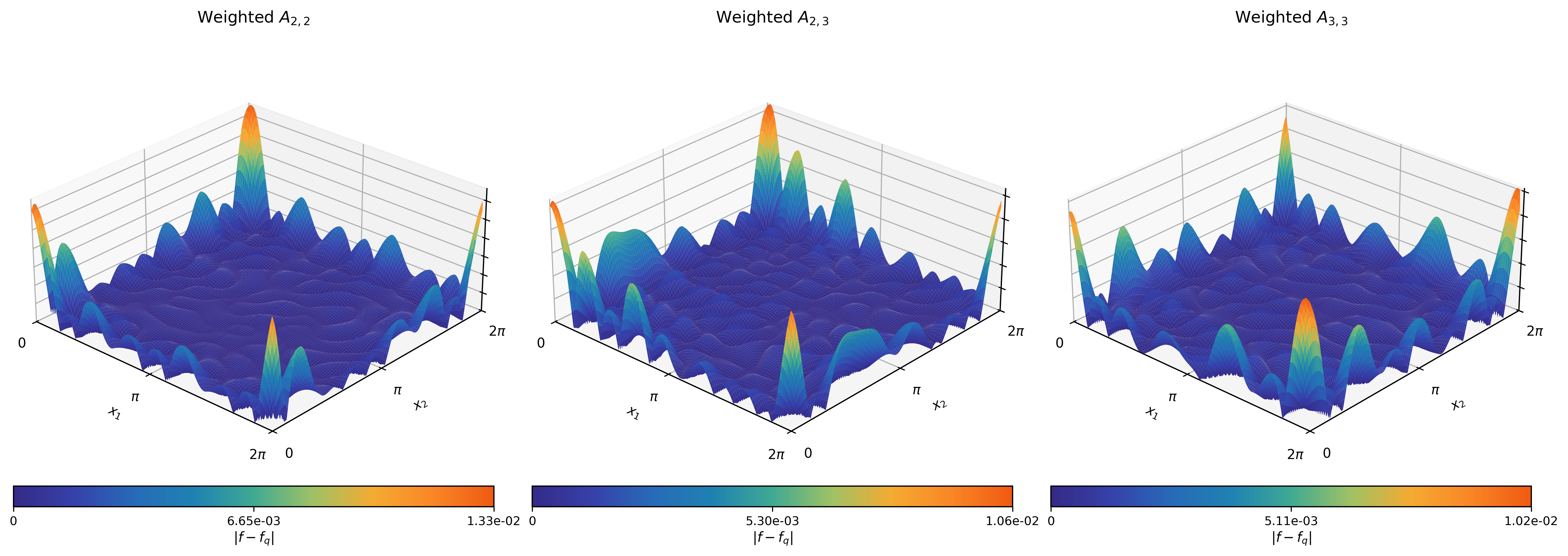}
\caption{Representative absolute reconstruction errors for some second-order weighted \(\Sigma\Delta\) schemes introduced in \cite{KrahmerVeselovska2023}. The error is amplified near the artificial seams of the Cartesian domain.}
\label{image_boundary1}
\end{figure}

In this paper, we show that this issue can be addressed by an alternative discretization, which is designed to be captured by a single trajectory.

\section{Rank-one lattice sampling and quantization}
\label{sec:rank-one-lattice}

The preceding discussion shows that restarting the one-dimensional quantizer
along several finite trajectories can produce a separate terminal boundary
contribution for each trajectory. We therefore replace these independently
initialized paths with a single ordering containing all the sampling points.
Rather than modifying the causal quantization rule, we modify the sampling
geometry and arrange the points along one rank-one lattice orbit on the torus.
The resulting samples form a uniformly parametrized one-dimensional sequence,
so that the $\Sigma\Delta$ recursion can be run only once and consequently
produces only one terminal boundary contribution.

To make this precise, let $M\in\mathbb N$, set $N:=M^2+1$ and introduce the uniformly spaced parameters
\[
t_n:=\frac{2\pi n}{N},
\qquad n\in\mathbb Z.
\]

Then the rank-one lattice generated by $z_M=(1,M)$ with modulus $N$ is obtained by considering (see \cite{potts2021approximation} for a generalized discussion on how to construct rank-one lattices)
\begin{equation}
\label{eq:lattice-orbit}
\gamma_n
:=
\left[
t_n
\begin{pmatrix}
1\\
M
\end{pmatrix}
\right]_{2\pi}
=
\left[
\begin{pmatrix}
2\pi n/N\\
2\pi Mn/N
\end{pmatrix}
\right]_{2\pi},
\qquad n\in\mathbb Z,
\end{equation}
where $[\cdot]_{2\pi}$ denotes the modulo operation, i.e., the remainder when dividing by $2\pi$. The name rank-one lattice stems from the fact that starting from $\gamma_0=(0,0)$, each point is obtained from the preceding one
by adding the fixed vector $\omega_M
:=
\frac{2\pi}{N}
\begin{pmatrix}
1\\
M
\end{pmatrix}
$
and reducing both coordinates modulo $2\pi$, i.e., 
\[
\gamma_{n+1}=[\gamma_n+\omega_M]_{2\pi}.
\]
It is well known that $\gamma_{n+N}=\gamma_n,$
so the orbit closes after exactly $N$ steps. Consequently, the samples $y_n:=f(\gamma_n), \qquad n=0,\ldots,N-1,$
can be processed as a single finite one-dimensional sequence.
The proposed orbit also retains, at the relevant spatial scale, the geometry
of a Cartesian raster. Indeed, the sampling points induce a partition of the torus into $M^2+1$ squares of side length $\tfrac{2\pi}{\sqrt{M^2+1}}$.

We emphasize that the trajectory is closed in the geometric sense, but the
quantizer will still be the usual causal finite recursion, rather than a cyclic
quantizer. However, the construction yields only a single boundary term, rather than one for each row and column.  not automatically eliminate the
remaining terminal boundary term; its purpose is to avoid producing one such
term for every separately initialized row or trajectory.

It remains to verify that these sampling points capture all the information
contained in a function in $P_K(\mathbb T^2)$. This follows from the standard reconstructing criterion for
rank-one lattices, see, for instance, \cite[Section~2.1]{potts2021approximation}, which is at the core of the following proposition. 

\begin{proposition}[Sampling and reconstruction along the lattice orbit]
\label{prop:spiral-sampling}
Let $K,M\in\mathbb N$, let $N=M^2+1$, and let $\Gamma_M:=\{\gamma_n \ : \ n = 0,\dots,N-1\}$ be defined by
\eqref{eq:lattice-orbit}. Then a function in $P_K$ is uniquely identified by its samples in $\Gamma_M$ if and only if 
\begin{equation}
\label{eq:M-condition}
M\geq 2K+1.
\end{equation}
Concretely, under this condition, every $f\in P_K(\mathbb T^2)$ can be reconstructed via
\begin{equation}
\label{eq:general-direct-reconstruction}
f(x)
=
\frac1N
\sum_{n=0}^{N-1}
f(\gamma_n)D_{K}(x-\gamma_n),
\qquad x\in\mathbb T^2,
\end{equation}
where
\[
D_{K}(x):=\sum_{k\in I_K}e^{\mathrm{i}k\cdot x}.
\]
\end{proposition}

\begin{proof}
The standard reconstructing criterion for rank-one lattices (see \cite[Section~2.1]{potts2021approximation}) states that a necessary and sufficient condition for the reconstruction formula $18...$ to hold is that the map
\[
k\longmapsto k\cdot z_M \pmod N
\]
is injective on $I_K:= [-K,K]^2\cap\mathbb{Z}^2$, or, equivalently,
\[
m\cdot z_M\not\equiv0\pmod N
\qquad
\text{for every }
m\in(I_K-I_K)\setminus\{0\}.
\]

Since $I_K-I_K=[-2K,2K]^2\cap\mathbb Z^2,$
the reconstructing criterion reduces to proving that

\begin{equation}
\label{eq:collision-congruence}
r+Ms\not\equiv0\pmod{M^2+1}
\end{equation}
for every nonzero \((r,s)\in[-2K,2K]^2\cap\mathbb Z^2\).
Suppose first that \(M\geq 2K+1\), and assume that
$$
r+Ms\equiv0\pmod{M^2+1}
$$
for some \((r,s)\in[-2K,2K]^2\cap\mathbb Z^2\). Since \(2K\leq M-1\), we have
$$
|r+Ms|
\leq |r|+M|s|
\leq 2K(M+1)
\leq (M-1)(M+1)
=M^2-1
<M^2+1.
$$
Thus, the only multiple of \(M^2+1\) that \(r+Ms\) can equal is zero, and hence

$$
r+Ms=0.
$$

If \(s\neq0\), this identity gives

$$
|r|=M|s|\geq M>2K,
$$

contradicting \(|r|\leq2K\). Therefore \(s=0\), and consequently \(r=0\). Hence no nonzero element of \(I_K-I_K\) satisfies the collision congruence, so the map

$$
k\longmapsto k\cdot(1,M)\pmod{M^2+1}
$$

is injective on \(I_K\).

Conversely, suppose that \(M\leq2K\). Then $(-M,1)\in (I_K-I_K)\setminus\{0\},$
and
$$
(-M,1)\cdot(1,M)=0\equiv0\pmod{M^2+1}.
$$
Therefore, the reconstruction criterion fails. 
\end{proof}

\subsection{Quantization along the rank-one lattice}
\label{sec:quantization-error-after-update}

We now combine the one-dimensional noise-shaping relation with the rank-one
lattice reconstruction formula. The crucial point is that the lattice nodes
form a single closed ordering. Hence, we obtain one global terminal
condition, rather than a collection of unrelated conditions associated with
different rows.

Let
\[
 N=M^2+1,
 \qquad
 t_n=\frac{2\pi n}{N},
 \qquad
 \gamma_n=[(t_n,Mt_n)]_{2\pi},
\]
and assume that $M\geq2K+1$. With this notation,
\eqref{eq:general-direct-reconstruction} represents $f(x)$,
$x=(x_1,x_2)\in\mathbb T^2$, as a linear combination of the
one-dimensional trajectory kernel
\begin{equation}
\label{eq:trajectory-kernel-error}
 \Psi_M(x;t)
 :=D_K\bigl(x-(t,Mt)\bigr)
 =\varphi^K(x_1-t)\varphi^K(x_2-Mt),
 \qquad t\in\mathbb T,
\end{equation}
evaluated at $t=t_n$ with the samples $y_n:=f(\gamma_n)$ as coefficients. The modulo operation can be omitted because $D_K$ is
periodic in both variables, so we obtain the representation
\begin{equation}
\label{eq:trajectory-reconstruction-error}
 f(x)=\frac1N\sum_{n=0}^{N-1}y_n\Psi_M(x;t_n).
\end{equation}

Our strategy is to quantize this representation, that is, we aim to find a sequence $(q_n)_{n = 0}^{ N-1}$ such that the quantized representation 
\begin{equation}\label{reconstructfromq}
    f_q(x):= \frac{1}{N}\sum_{n = 0}^{N-1}q_n\Psi_M(x;t_n)
\end{equation}
approximates $f$, that is,
$$e_N(x) := f(x) - f_q(x) =\frac1N\sum_{n=0}^{N-1}(y_n-q_n)\Psi_M(x;t_n)$$
is small in the supremum norm. 

To find the $q_n$'s, we apply $\Sigma\Delta$ modulation along the lattice. This gives rise to a single one-dimensional
recursion, which we initialize by $0$ for $n<0$.  
We proceed by analyzing the bounds for the reconstruction error for the first and second order quantization in the following theorem. The proof structure is similar to Proposition $2$ in \cite{GrafKrahmerKrauseSolberg2023}, but adapted to a $2$-dimensional kernel decomposition; we include the proof for completeness.

\begin{theorem}
\label{prop:spiral-quantization-errors}
Let \(f\in P_K(\mathbb T^2)\), $ \|f\|_{L^\infty(\mathbb T^2)}\leq\mu<1.$ Let $M\ge 2K+1$, $N = M^2+1$ and $y_n := f(\gamma_n)$, where $(\gamma_n)_{n = 0}^{N-1}$ is the rank-one lattice defined by \eqref{eq:lattice-orbit}. Then the following statements hold.

$i)$ \textbf{First-order quantization.}\\
If the sequence $(y_n)$ is quantized via 
\begin{equation}\left\{
    \begin{aligned}
        &u_n = u_{n-1} +y_n -q_n\\
        & q_n = \operatorname{sign}(u_{n-1} + y_n)\\
        & u_{-1} = 0
    \end{aligned}\right., 
\end{equation}
then 
\begin{equation}
\label{eq:first-order-spiral-error-bound}
 |f(x)-f_q(x)|
 \leq
 \frac1N
 \left(
 \|\Psi_M(x;\cdot)\|_{L^\infty(\mathbb T)}
 +
 \|\partial_t\Psi_M(x;\cdot)\|_{L^1(\mathbb T)}
 \right).
\end{equation}
$ii)$ \textbf{Second-order quantization.}
Let $h^{(k)}$ be a second-order minimally supported filter as in \eqref{second_order_minimal}, with $k \ge \lceil2/(1-\mu)\rceil$. If the sequence $(y_n)$ is quantized via
\begin{equation}\left\{
    \begin{aligned}
        &v_n = \frac{k+1}{k}v_{n-1} - \frac{1}{k}v_{n-(k+1)} +y_n -q_n\\
        & q_n = \operatorname{sign}\left(\frac{k+1}{k}v_{n-1} - \frac{1}{k}v_{n-(k+1)} + y_n\right)\\
        &v_{-i} = 0 \ \text{for all } i<0
    \end{aligned}\right., 
\end{equation}
then 
\begin{equation}
\label{eq:second-order-spiral-error-bound}
\begin{aligned}
 |f(x)-f_q(x)|
 &\leq
 \frac{k+1}{N}
 \|\Psi_M(x;\cdot)\|_{L^\infty(\mathbb T)}
 \\
 &\quad+
 \frac{\pi(k+1)}{N^2}
 \left(
 \|\partial_t\Psi_M(x;\cdot)\|_{L^\infty(\mathbb T)}
 +
 \|\partial_t^2\Psi_M(x;\cdot)\|_{L^1(\mathbb T)}
 \right).
\end{aligned}
\end{equation}
\end{theorem}

\begin{proof}
Fix \(x\in\mathbb T^2\) and set $\psi_n:=\Psi_M(x;t_n), \ 
 \tau:=\frac{2\pi}{N}.$
Notice that \(t\mapsto\Psi_M(x;t)\) is a smooth \(2\pi\)-periodic
function. Thus, the modulo operation appearing in the definition of the
lattice produces no discontinuity in the following integral estimates.

\medskip
\noindent
$i)$ \textbf{First-order scheme.}
The recursion gives
\[
 y_n-q_n=u_n-u_{n-1}=(\Delta u)_n.
\]
Moreover, since \(|y_n|\leq\mu<1\) with \(u_{-1}=0\), from \eqref{eq:gunturk-stability-criterion} one has $ \|u\|_{\ell^\infty}\leq1$

Discrete summation by parts yields
\begin{align}
 f(x) -f_q(x) = \frac{1}{N}\sum_{n=0}^{N-1}(\Delta u)_n\psi_n
 &=\frac{1}{N}\left(
 u_{N-1}\psi_{N-1}
 +
 \sum_{n=0}^{N-2}u_n(\psi_n-\psi_{n+1})\right).
\label{eq:first-order-sbp}
\end{align}
For \(n=0,\ldots,N-2\), the fundamental theorem of calculus gives
\[
 \psi_{n+1}-\psi_n
 =
 \int_{t_n}^{t_{n+1}}
 \partial_t\Psi_M(x;t)\,dt.
\]
Consequently,
\begin{align*}
 \sum_{n=0}^{N-2}|\psi_{n+1}-\psi_n|
 &\leq
 \sum_{n=0}^{N-2}
 \int_{t_n}^{t_{n+1}}
 |\partial_t\Psi_M(x;t)|\,dt \leq
 \|\partial_t\Psi_M(x;\cdot)\|_{L^1(\mathbb T)}.
\end{align*}
Combining the above estimate with \(\|u\|_{\ell^\infty}\leq1\), one obtains
\[
 |f(x)-f_q(x)|
 \leq
 \frac1N
 \left(
 \|\Psi_M(x;\cdot)\|_{L^\infty(\mathbb T)}
 +
 \|\partial_t\Psi_M(x;\cdot)\|_{L^1(\mathbb T)}
 \right).
\]

\medskip
\noindent
$ii)$ \textbf{Second-order scheme.}
The assumed stability condition for \(h^{(k)}\), together with the zero
initial conditions, implies $\|v\|_{\ell^\infty}\leq1$ (see Section $2$).
Therefore, recalling formula \eqref{eq:generalized-order-condition} with $u:=g^{(k)}*v$, it holds that
\[
 y-q
 =
 (\delta_0-h^{(k)})*v
 =
 \Delta^2(g^{(k)}*v)
 =
 \Delta^2u.
\]
and $ \|u\|_{\ell^\infty}\leq
 \|g^{(k)}\|_{\ell^1}= (k+1)/2$, where the last equality follows from the fact that $\|g\|_{\ell^1}$ is given by the product of support locations of the filter $h$ divided by $2$ (see \cite{Gunturk2003}). \\
Similarly to the argument in $i)$, applying discrete summation by parts twice, and using the fact that with the chosen initialization of $v$ one has \(u_{-1}=u_{-2}= 0\), gives
\begin{align}
 f(x)-f_q(x) = \frac{1}{N}\sum_{n=0}^{N-1}(\Delta^2u)_n\psi_n
 &=
  \frac{1}{N}(\Delta u)_{N-1}\psi_{N-1}
 +
 \frac{1}{N} u_{N-2}(\psi_{N-2}-\psi_{N-1})\\&+ \frac{1}{N}\sum_{n=0}^{N-3}u_n(\psi_n-2\psi_{n+1}+\psi_{n+2}).
\label{eq:second-order-sbp}
\end{align}
The first boundary term satisfies
\[
 |(\Delta u)_{N-1}\psi_{N-1}|
 \leq
 2\|u\|_{\ell^\infty}
 \|\Psi_M(x;\cdot)\|_{L^\infty(\mathbb T)},
\]
while for the second boundary term it holds that
\[
 |\psi_{N-2}-\psi_{N-1}|
 \leq
 \tau
 \|\partial_t\Psi_M(x;\cdot)\|_{L^\infty(\mathbb T)}.
\]

For the interior second differences, the fundamental theorem of calculus
applied twice gives
\begin{align*}
 \psi_n-2\psi_{n+1}+\psi_{n+2}
 &=
 \int_{t_n}^{t_{n+1}}
 \int_s^{s+\tau}
 \partial_t^2\Psi_M(x;r)\,dr\,ds.
\end{align*}
Therefore, using the periodicity of
\(\partial_t^2\Psi_M(x;\cdot)\),
\begin{align*}
 &\sum_{n=0}^{N-3}
 |\psi_n-2\psi_{n+1}+\psi_{n+2}| \leq
 \sum_{n=0}^{N-1}
 \int_{t_n}^{t_{n+1}}
 \int_s^{s+\tau}
 |\partial_t^2\Psi_M(x;r)|\,dr\,ds\\&=
 \int_0^{2\pi}
 \int_s^{s+\tau}
 |\partial_t^2\Psi_M(x;r)|\,dr\,ds =
 \tau
 \|\partial_t^2\Psi_M(x;\cdot)\|_{L^1(\mathbb T)}.
\end{align*}
The last equality follows from Fubini's theorem and periodicity.

Using these estimates in \eqref{eq:second-order-sbp}, we obtain
\begin{align*}
 \left|
 \sum_{n=0}^{N-1}(\Delta^2u)_n\psi_n
 \right|
 &\leq2\|u\|_{\ell^\infty}
 \|\Psi_M(x;\cdot)\|_{L^\infty(\mathbb T)}
\\&+
 \tau\|u\|_{\ell^\infty}
 \left(
 \|\partial_t\Psi_M(x;\cdot)\|_{L^\infty(\mathbb T)}
  +
 \|\partial_t^2\Psi_M(x;\cdot)\|_{L^1(\mathbb T)}
 \right).
\end{align*}

Finally, recalling that \(\tau=2\pi/N\), we conclude
that
\begin{align*}
 |f(x)-f_q(x)|
 &\leq
 \frac{2\|g^{(k)}\|_{\ell^1}}{N}
 \|\Psi_M(x;\cdot)\|_{L^\infty(\mathbb T)}
 \\
 &\quad+
 \frac{2\pi\|g^{(k)}\|_{\ell^1}}{N^2}
 \left(
 \|\partial_t\Psi_M(x;\cdot)\|_{L^\infty(\mathbb T)}
 +
 \|\partial_t^2\Psi_M(x;\cdot)\|_{L^1(\mathbb T)}
 \right),
\end{align*}
which proves the result.
\end{proof}

\begin{remark}
Observe that second-order quantization requires a constraint on the signal amplitude. In
digital halftoning, this may be problematic, since the image intensities must
be rescaled before quantization to satisfy the stability criterion. The
resulting halftone would then represent a normalized version of the image
rather than its original intensity range. This limitation can be mitigated by
using feedback filters with longer support, although the required support may
become prohibitively large as the signal amplitude approaches the stability
threshold. Recent work \cite{joy2026trajectory} shows, however, that the
classical sufficient condition for minimally supported second-order filters is
often overly conservative. Under suitable smoothness assumptions on the input
function, stability can still be guaranteed for signals with substantially
larger amplitudes. In particular, for a filter of length $k+1$, if the function is smooth enough, it is possible to quantize a signal of maximum amplitude up to $1-e^2/(k+1)^2$. We believe that the framework of this paper can be adapted to this more complicated stability analysis, but as our focus is on the boundary effects, we decided to base our results on the simpler guarantees.
\end{remark}

To obtain explicit decay rates from
Proposition~\ref{prop:spiral-quantization-errors}, it remains to estimate the
norms of the trajectory kernel and its derivatives.

\begin{lemma}
\label{lem:trajectory-kernel-estimates}
For \(j=0,1,2\), set
\begin{equation}
\label{eq:dirichlet-kernel-norms}
 A_j
 :=
 \bigl\|(\varphi^K)^{(j)}\bigr\|_{L^\infty(\mathbb T)},
 \qquad
 B_j
 :=
 \bigl\|(\varphi^K)^{(j)}\bigr\|_{L^1(\mathbb T)}.
\end{equation}
It holds that
\begin{align}
 \|\Psi_M(x;\cdot)\|_{L^\infty(\mathbb T)}
 &\leq A_0^2 = O(1),
 \label{eq:trajectory-kernel-Linf}\\
 \|\partial_t\Psi_M(x;\cdot)\|_{L^1(\mathbb T)}
 &\leq A_1B_0+MA_0B_1 = O(M),
 \label{eq:trajectory-first-derivative-L1}\\
 \|\partial_t\Psi_M(x;\cdot)\|_{L^\infty(\mathbb T)}
 &\leq (M+1)A_0A_1= O(M),
 \label{eq:trajectory-first-derivative-Linf}\\
 \|\partial_t^2\Psi_M(x;\cdot)\|_{L^1(\mathbb T)}
 &\leq
 A_2B_0+2MA_1B_1+M^2A_0B_2= O(M^2).
 \label{eq:trajectory-second-derivative-L1}
\end{align}
\end{lemma}

\begin{proof}
The product and chain rules give
\begin{align}
 \partial_t\Psi_M(x;t)
 &=
 -(\varphi^K)'(x_1-t)\varphi^K(x_2-Mt)
 \notag\\
 &\quad
 -M\varphi^K(x_1-t)(\varphi^K)'(x_2-Mt),
 \label{eq:trajectory-first-derivative}\\
 \partial_t^2\Psi_M(x;t)
 &=
 (\varphi^K)''(x_1-t)\varphi^K(x_2-Mt)
 \notag\\
 &\quad
 +2M(\varphi^K)'(x_1-t)(\varphi^K)'(x_2-Mt)
 \notag\\
 &\quad
 +M^2\varphi^K(x_1-t)(\varphi^K)''(x_2-Mt).
 \label{eq:trajectory-second-derivative}
\end{align}

Because \(M\) is an integer, the map \(t\mapsto x_2-Mt\) traverses
exactly \(M\) complete periods as \(t\) ranges over \([0,2\pi]\).
Consequently, for \(j=0,1,2\),
\begin{align}
 \int_0^{2\pi}
 \bigl|(\varphi^K)^{(j)}(x_2-Mt)\bigr|\,dt
 &=
 \frac1M
 \int_{x_2-2\pi M}^{x_2}
 \bigl|(\varphi^K)^{(j)}(s)\bigr|\,ds
 \notag\\
 &=
 B_j.
 \label{eq:periodic-change-of-variables}
\end{align}

The estimate \eqref{eq:trajectory-kernel-Linf} follows by bounding both
factors in \(L^\infty\). Applying the triangle inequality to
\eqref{eq:trajectory-first-derivative}, estimating one factor in
\(L^\infty\) and the other in \(L^1\), gives
\[
 \|\partial_t\Psi_M(x;\cdot)\|_{L^1}
 \leq
 A_1B_0+MA_0B_1.
\]
Bounding both factors in \(L^\infty\) gives
\[
 \|\partial_t\Psi_M(x;\cdot)\|_{L^\infty}
 \leq
 A_1A_0+MA_0A_1
 =
 (M+1)A_0A_1.
\]
Applying the same \(L^\infty\)-\(L^1\) argument to
\eqref{eq:trajectory-second-derivative} yields
\[
 \|\partial_t^2\Psi_M(x;\cdot)\|_{L^1}
 \leq
 A_2B_0+2MA_1B_1+M^2A_0B_2.
\]

Since \(A_j\) and \(B_j\) depend only on \(K\), the asymptotic estimates follow.
\end{proof}

\begin{corollary}[First- and second-order decay rates]
\label{cor:spiral-quantization-decay}
Under the assumptions of
Proposition~\ref{prop:spiral-quantization-errors}, consider the quantized representation $f^{(1)}_q$ and $f^{(2)}_q$ of $f$, respectively for the first and second order $\Sigma\Delta$ schemes.
For a fixed bandwidth \(K\), the following asymptotic estimates hold,
\begin{align}
 \|f-f_q^{(1)}\|_{L^\infty(\mathbb T^2)}
 &=
 O(M^{-1})
 =
 O(N^{-1/2}),
 \label{eq:first-order-spiral-final-decay}\\
 \|f-f_q^{(2)}\|_{L^\infty(\mathbb T^2)}
 &=
 O(M^{-2})
 =
 O(N^{-1}).
 \label{eq:second-order-spiral-final-decay}
\end{align}
\end{corollary}

\begin{proof}
The statement holds by plugging in the estimates in Lemma \ref{lem:trajectory-kernel-estimates} into the error bounds of Theorem \ref{prop:spiral-quantization-errors}.
\end{proof}

\begin{remark}[The global update and its role on the rank-one lattice]
\label{rem:role-global-update}
The global update was introduced in
\cite{GrafKrahmerKrauseSolberg2023} to handle the terminal mismatch produced
by a causal \(\Sigma\Delta\) quantizer on a periodic domain. Although the
samples lie on a circle, the recursion is initialized at one point and runs
causally; hence, the final state does not generally match the initial state.

The idea is to perturb every input sample by the same constant and then rerun
the quantizer. This gives rise to a better approximation of a shifted function, so the order of the error remains the same, but its leading term is now globally constant. This is preferred from the perspective of digital halftoning as changing an image by a small global constant is less disruptive than having a localized error of comparable size. To make this precise for the second-order case with a minimally supported
filter \(h^{(k)}\), as introduced in \eqref{second_order_minimal}, denote by
\(u^{(2)}=g^{(k)}*v^{(2)}\) the associated second-order state satisfying
\(y-q=\Delta^2u^{(2)}\). Starting from the states generated by the original
samples, the first- and second-order updates are defined by
\begin{equation}
\label{global_updates}
 \delta^{(1)}
 :=
 -\frac{u_{N-1}^{(1)}}{N},
 \qquad
 \delta^{(2)}
 :=
 -\frac{(\Delta u^{(2)})_{N-1}}{N}.
\end{equation}
The quantizer is then applied again to
\[
 \widetilde y_n^{(r)}
 :=
 y_n+\delta^{(r)},
 \qquad n=0,\ldots,N-1.
\]
Under the assumptions of
\cite[Theorems~1 and~2]{GrafKrahmerKrauseSolberg2023}, the updated states
satisfy
\[
 \widetilde u_{N-1}^{(1)}=0,
 \qquad
 (\Delta\widetilde u^{(2)})_{N-1}=0.
\]

In the one-dimensional periodic setting of
\cite{GrafKrahmerKrauseSolberg2023}, the second-order terminal contribution is
\(O(N^{-1})\), whereas the interior term is \(O(N^{-2})\). The global update
removes the terminal term and therefore improves the overall decay from
\(O(N^{-1})\) to \(O(N^{-2})\).

The same construction can be applied to the present rank-one lattice because
all \(N=M^2+1\) samples form a single causal trajectory. In particular, only
one constant update is needed for the entire image. This would not be as easy
for independently quantized Cartesian rows, since each row generally
requires a different terminal correction.

For the first-order rank-one lattice scheme, the update cancels
\[
 \frac{u_{N-1}^{(1)}}{N}D_K(x-\gamma_0),
\]
which is \(O_K(N^{-1})=O_K(M^{-2})\). The remaining interior error is
\(O_K(M/N)=O_K(M^{-1})\), so the update does not change the first-order decay
rate.

For the second-order spiral scheme, it cancels
\[
 \frac{(\Delta u^{(2)})_{N-1}}{N}D_K(x-\gamma_0),
\]
which is again \(O_K(N^{-1})=O_K(M^{-2})\). However, the interior contribution
already satisfies
\[
 \frac{1}{N^2}
 \|\partial_t^2\Psi_M(x;\cdot)\|_{L^1}
 =
 O_K\left(\frac{M^2}{N^2}\right)
 =
 O_K(M^{-2}).
\]
Therefore, unlike the one-dimensional result, the update does not improve the order of the error decay.

Nevertheless, our numerical experiments show that there's still a benefit in terms of the reconstruction error, even if it does not reflect in the asymptotic order. 
\end{remark}

\section{Faster error decay for non-squared pixels}
\label{sec:fixed-direction-spirals}

The rank-one lattice described above uses \(M^2+1\) square pixels, and thus the direction of the vector $\begin{pmatrix}
1\\
M
\end{pmatrix}$ generating the lattice changes with the sampling density.  This fact enters directly into the analysis of quantization error, saturating, even after the constant update, its decay rate to $O(1/M)$ or $O(1/M^2)$, respectively, for a first and second order quantization schemes. To obtain a better decay rate, one can instead keep the vector $\begin{pmatrix}
1\\
S
\end{pmatrix}$ generating the lattice fixed and increase only the number of samples along the trajectory. Note that, as a result, the pixels in the halftoned representation of an image will no longer be squares, but rather rectangles, see Figure \ref{uniformvsnonuniform}.

\begin{figure}
    \centering
    \includegraphics[width=0.8\linewidth]{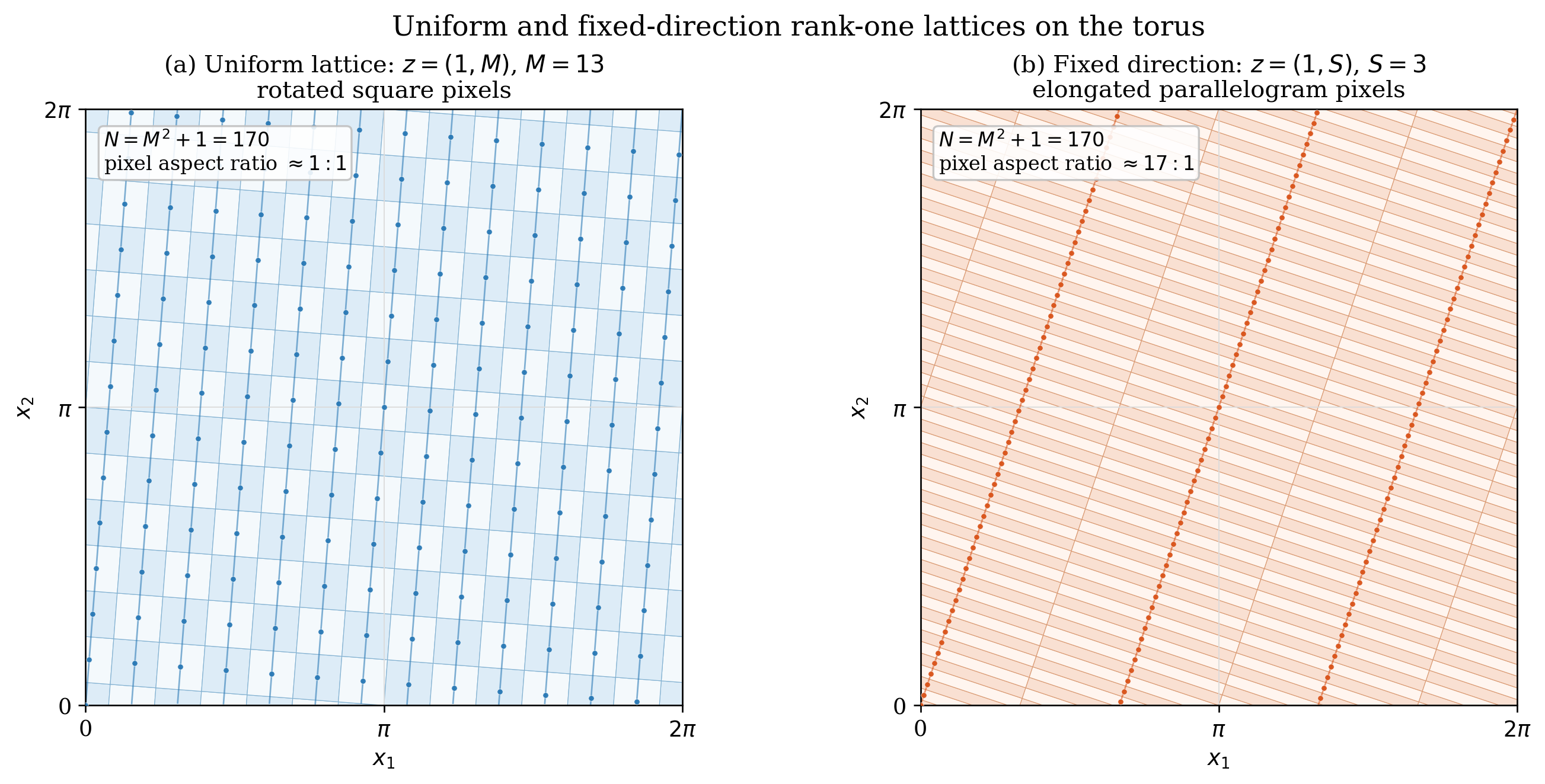}
    \caption{Comparison between the uniform and fixed-direction rank-one
lattices on \(\mathbb T^2\), using the same number
\(N=M^2+1=170\) of sampling points.}
    \label{uniformvsnonuniform}
\end{figure}

More precisely, fix an integer \(S\geq 2K+1\) and, for every
\(M\in\mathbb N\) satisfying $ M^2+1>2K(S+1)$, consider the sampling points
\[
 \gamma_n^{S,M}
 :=
 \left[
 \frac{2\pi n}{M^2+1}
 \begin{pmatrix}
 1\\ S
 \end{pmatrix}
 \right]_{2\pi},
 \qquad n=0,\ldots,M^2.
\]
By the standard injectivity criterion for rank-one lattices
(see, for instance, \cite[Section~2.1]{potts2021approximation}), the condition above ensures that
the map
\[
 k=(k_1,k_2)
 \longmapsto
 k_1+Sk_2
 \pmod{M^2+1}
\]
is injective on \(I_K=[-K,K]^2\cap\mathbb Z^2\). Therefore, every
\(f\in P_K(\mathbb T^2)\) can be reconstructed from the samples
\(f(\gamma_n^{S,M})\) according to
\begin{equation}
\label{eq:fixed-direction-reconstruction-formula}
\begin{aligned}
 f(x)
 &=
 \frac1{M^2+1}
 \sum_{n=0}^{M^2}
 f(\gamma_n^{S,M})
 D_K(x-\gamma_n^{S,M})\\
\end{aligned}
\end{equation}

Following the argument of Proposition
\ref{prop:spiral-quantization-errors}, a \(\Sigma\Delta\) scheme applied to
the ordered samples \(f(\gamma_n^{S,M})\) produces boundary terms of order
\(O(M^{-2})\), independently of whether a first- or second-order scheme is
used. Consequently, without an update, the boundary contribution limits the
decay rate to \(O(M^{-2})\) in both cases. However, if $S$ is fixed, the constant update cancels this
contribution: it preserves the \(O(M^{-2})\) rate for the first-order scheme,
while improving the second-order rate from \(O(M^{-2})\) to \(O(M^{-4})\).
We now state the corresponding quantization-error estimate.

\begin{theorem}[Quantization along a fixed-direction rank-one lattice]
\label{thm:fixed-direction-error-decay}
Let \(f\in P_K(\mathbb T^2)\), $ \|f\|_{L^\infty(\mathbb T^2)}\leq\mu<1.$
Fix an integer \(S\geq2K+1\), and let \(M\in\mathbb N\) be such that $M^2>2K(S+1).$
Consider the samples
\[
 y_n:=f(\gamma_n^{S,M}),
 \qquad n=0,\ldots,M^2-1.
\]

Under the assumptions of
Theorem~\ref{prop:spiral-quantization-errors}, consider a first run of the $\Sigma\Delta$ schemes and let
\(\delta^{(r)}\) be the global updates as in \eqref{global_updates}, and set
\[
 \widetilde y_n^{(r)}:=y_n+\delta^{(r)},
 \qquad r=1,2.
\]

Then the following statements hold.

\medskip
\noindent
\textnormal{\(i\)) \textbf{First-order quantization.}}
Suppose that \(\widetilde y^{(1)}\) is quantized by the first-order scheme
of Theorem~\ref{prop:spiral-quantization-errors}, producing the bit sequence
\(\widetilde q^{(1)}\). Then there exists a constant \(C_{K,S}>0\),
independent of \(M\), such that
\begin{equation}
\label{eq:fixed-direction-first-order-rate}
 \bigl\|f+\delta^{(1)}- f_{\widetilde q}^{(1)}
 \bigr\|_{L^\infty(\mathbb T^2)}
 \leq
 \frac{C_{K,S}}{M^2}.
\end{equation}

\medskip
\noindent
\textnormal{\(ii\)) \textbf{Second-order quantization.}}
Suppose that \(\widetilde y^{(2)}\) is quantized by the corresponding
second-order scheme of
Theorem~\ref{prop:spiral-quantization-errors}, producing the bit sequence
\(\widetilde q^{(2)}\), together with the stability condition $$ \|\tilde{y}^{(2)}\|_{\ell^\infty}+\|h^{(k)}\|_{\ell^1}
 \leq2.$$ Then there exists a constant \(C_{K,S,k}>0\),
independent of \(M\), such that
\begin{equation}
\label{eq:fixed-direction-second-order-rate}
 \bigl\|f+\delta^{(2)}- f_{\widetilde q}^{(2)}
 \bigr\|_{L^\infty(\mathbb T^2)}
 \leq
 \frac{C_{K,S,k}}{M^4}.
\end{equation}
\end{theorem}

\begin{proof}
The proof follows exactly the argument used for the uniform rank-one lattice, with generator $\begin{pmatrix}
1\\
M
\end{pmatrix}$ and
\(\Psi_M\) replaced by
\(\Psi_S(x;t)=\varphi^K(x_1-t)\varphi^K(x_2-St)\). It also uses the fact that the quantization on the updated samples yields $\tilde{u}_{N-1}=0$ and $\Delta\tilde{u}_{N-1} =0$, respectively for the first and the second order schemes, see \cite{GrafKrahmerKrauseSolberg2023} for details.
\end{proof}

\section{Numerical experiments}
\label{sec:numerical-experiments}

We compare the Cartesian weighted second-order scheme associated with the scheme 
$A_{3,3}$ (see \cite{KrahmerVeselovska2023}) with the second-order schemes applied along the uniform rank-one
lattice and along the fixed direction.  In both rank-one lattice constructions,
the global update is applied before the final reconstruction.  We take $K=5$
and $M=256$; for the fixed-direction construction, we choose
$S=11=2K+1$, which satisfies the reconstruction condition.

More precisely, on the Cartesian grid, we use the weighted scheme $A_{3,3}$,
implemented as
\[
\begin{aligned}
 a_{m,n}
 &:={y}_{m,n}
 +\frac12\left(\frac43v_{m,n-1}-\frac13v_{m,n-4}\right)
 +\frac12\left(\frac43v_{m-1,n}-\frac13v_{m-4,n}\right),\\
 q_{m,n}&:=\operatorname{sign}(a_{m,n}),
 \qquad u_{m,n}:=a_{m,n}-q_{m,n},
\end{aligned}
\]
with initial conditions initialized to $0$.  Along each rank-one lattice, the
ordered samples are quantized using the same minimally supported filter
\[
 h^{(3)}=\frac43\delta_1-\frac13\delta_4,
 \qquad \|h^{(3)}\|_{\ell^1}=\frac53,
\]
and all test functions are normalized so that the
stability condition
\[
 \|y\|_{\ell^\infty}+|\delta|+
 \|h^{(3)}\|_{\ell^1}\leq2
\]
holds.
For $N$ samples on the trajectory, we first run this
recursion on $(y_n)_{n = 0}^{N-1}$ and record 
\[
 \delta=-\frac{(\Delta u)_{N-1}}{N},
\]
then we rerun it on the updated samples $y_n+\delta$.  The reconstruction obtained
from the updated bits is finally corrected by subtracting $\delta$.  This
procedure is used for both the uniform and the fixed-direction rank-one lattice.

\subsection{Average distribution of the error}
We first examine the spatial localization of the error over an ensemble of
$400$ randomly generated bandlimited functions.  At each point $x\in\mathbb
T^2$, we compute the pointwise ensemble RMS error
\[
 E_{\mathrm{RMS}}(x)
 :=
 \left(
 \frac1{400}\sum_{j=1}^{400}
 \bigl|f_j(x)-f_{q,j}(x)\bigr|^2
 \right)^{1/2}.
\]

After centering the torus for
visualization, one can see in Figure \ref{fig:ensemble-error-comparison} that $A_{3,3}$ produces two pronounced ridges along the row and column
where its recursions are initialized. Instead, the error of the updated
rank-one lattices is distributed without a preferred row or column. However, before the
update, the quantized rank-one values, in both the uniform and fixed direction lattices, produce an error near the closing point of the trajectory, compatible with the fact that the initial boundary term has not been eliminated. Observe that, as expected from theory, the approximation via the fixed direction rank-one lattice yields smaller reconstruction errors.

\begin{figure}[H]
    \centering
    \includegraphics[width=1\linewidth]{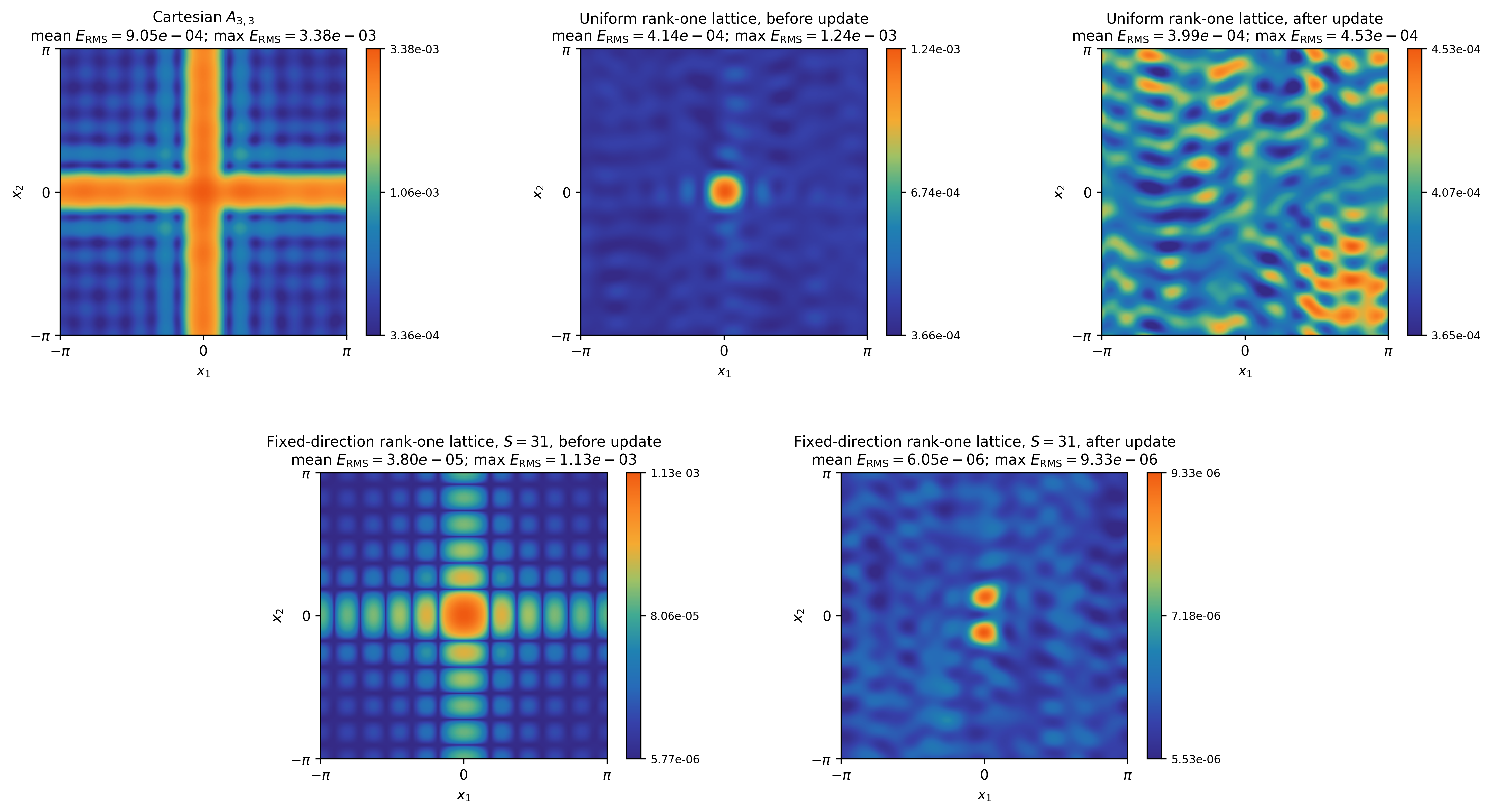}
    \caption{Pointwise ensemble RMS reconstruction error for the $A_{3,3}$ scheme, the uniform rank-one lattice before and after the constant update, as well as for the fixed direction rank-one lattice.  The error functions are shifted such that the $0$ is at the center of the pictures.}
    \label{fig:ensemble-error-comparison}
\end{figure}

\subsection{Comparison on a single function}
We next consider one nonconstant function from the same bandlimited test
family.  Figure~\ref{fig:single-function-five-errors} shows the pointwise
absolute reconstruction errors in log scale and for the different schemes.  The Cartesian reconstruction retains the
axis-aligned cross, whereas neither spiral reconstruction exhibits ridges
extending across the torus. Overall, the error for the rank-one lattice schemes, the constant update reduces the reconstruction error of the updated signals. In particular, for the rank-one lattice with fixed direction, the boundary artifact in the origin is eliminated in the updated scheme.

\begin{figure}[H]
    \centering
    \includegraphics[width=0.5\linewidth]{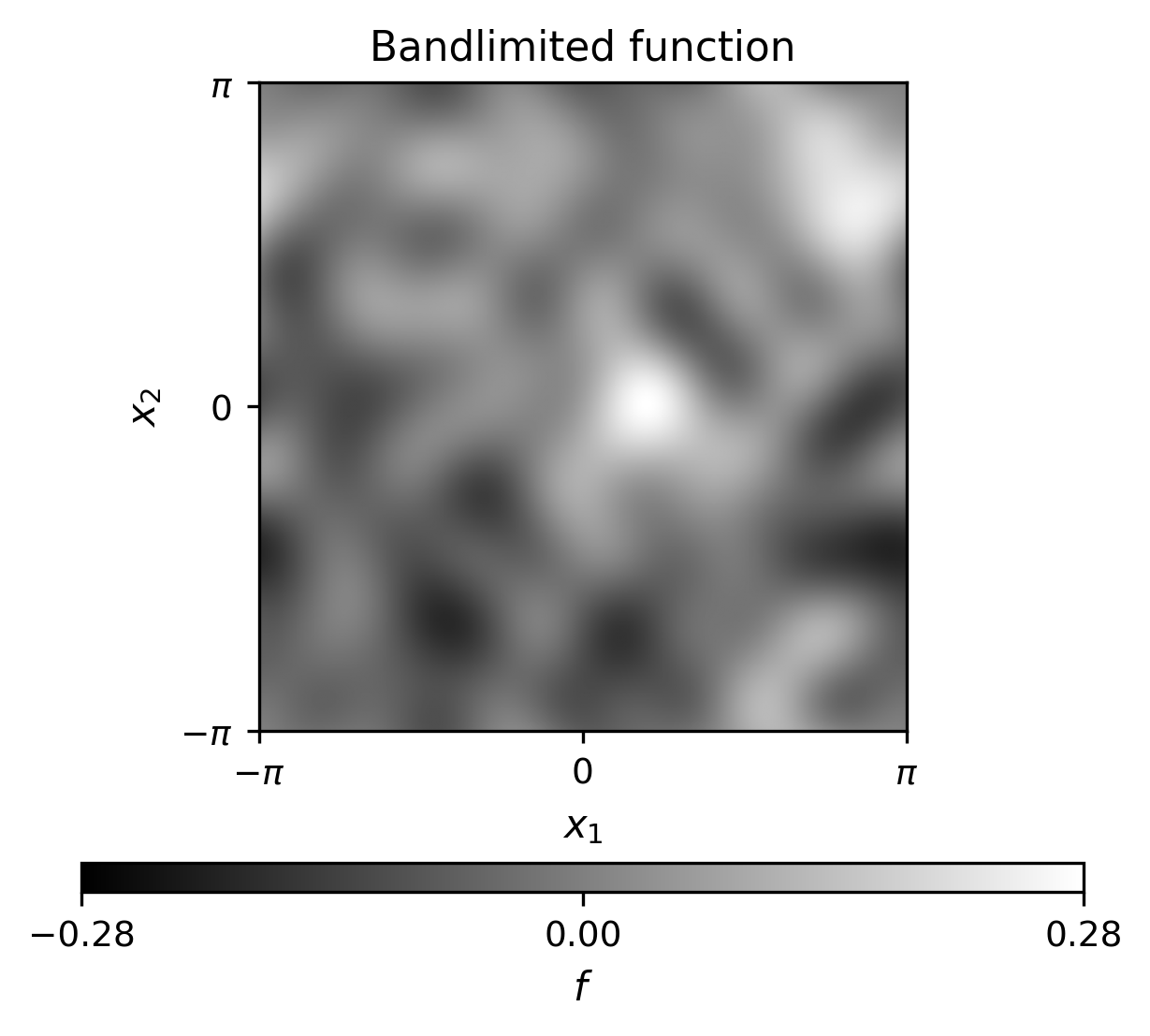}
    \caption{A real-valued bandlimited test function and the corresponding
    absolute reconstruction errors. The function is a randomly generated
    trigonometric polynomial in \(P_5(\mathbb T^2)\), with conjugate-symmetric
    Fourier coefficients supported on \([-5,5]^2\) and decreasing spectral
    weights. It is rescaled so that
    \(\|f\|_{L^\infty(\mathbb T^2)}=0.28\).}
    \label{fig:placeholder}
\end{figure}

\begin{figure}[H]
    \centering
    \includegraphics[width=\textwidth]{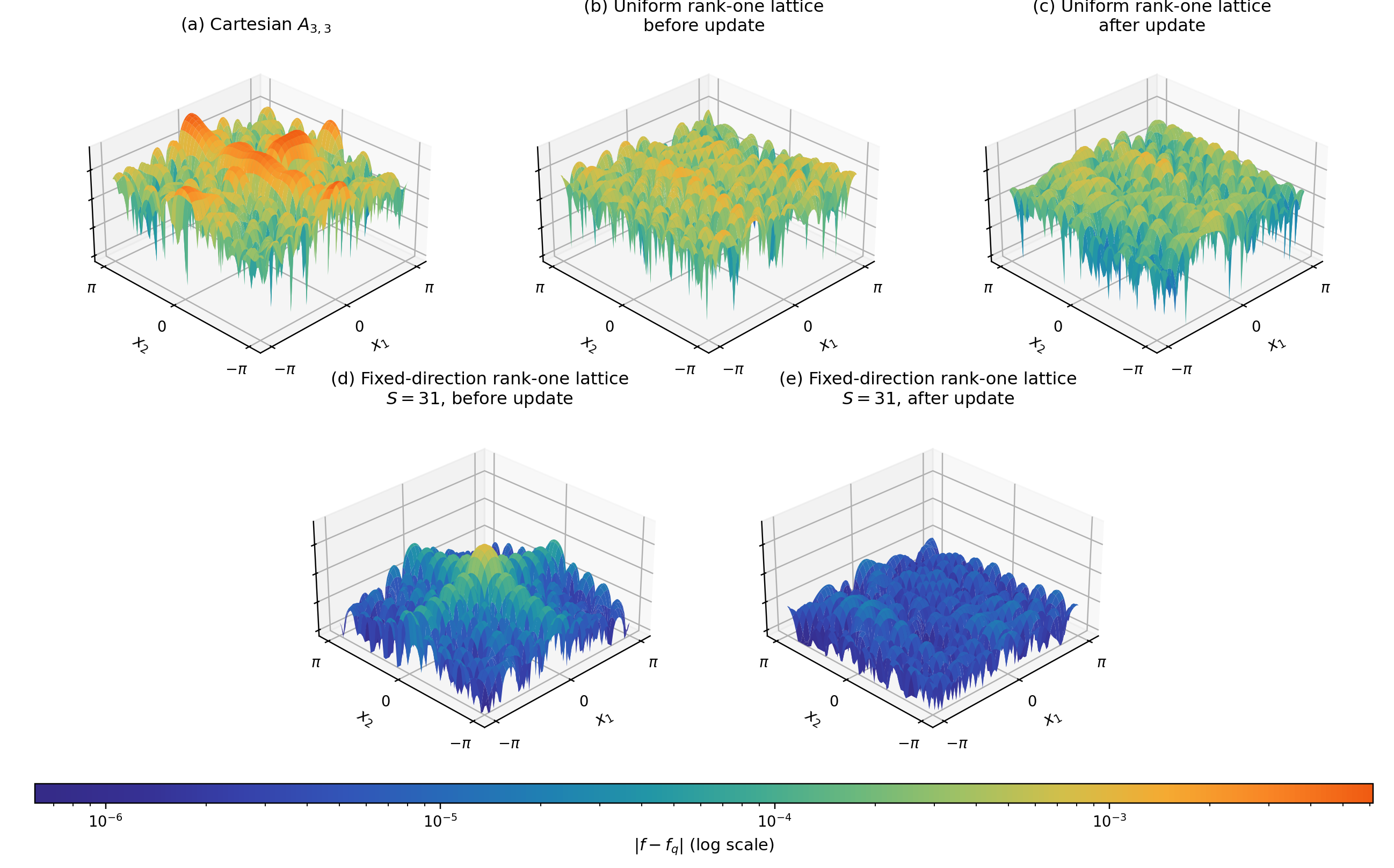}

    \caption{Absolute
    reconstruction errors obtained with the Cartesian weighted scheme
    \(A_{3,3}\) and with the uniform and fixed-direction rank-one lattices,
    before and after the global update. We use \(M=256\), the minimally
    supported second-order filter \(h^{(3)}\), and \(S=31\) for the
    fixed-direction lattice. A periodic shift moves the boundary to the
    center of each error surface. All five surfaces share the same
    logarithmic vertical and color scales.}
    \label{fig:single-function-five-errors}
\end{figure}

\subsection{An example of a halftoned periodic image}
\label{sec:flower-a88-example}

We consider the Shepp–Logan phantom, which is a periodic image, and we scale its gray values (in $[0,255]$)  such that after it is mapped to
$[-1,1]$ by $f=2I/255-1$, the resulting range is contained in
\((-1/2,1/2)\) so that no instability can arise due to large signal amplitude. The high-resolution periodic image is treated as a
continuous input through periodic bilinear interpolation, applying a
preliminary low-pass filter. We take $M=250$ and reconstruct the image using three different methods. The Cartesian experiment uses the
weighted scheme $A_{4,4}$ from \cite{KrahmerVeselovska2023}, i.e., the filter $h^{(4)}=\frac{5}{4}\delta_1-\frac{1}{4}\delta_5$ is applied on both the vertical and horizontal direction with both weights $1/2$ and $0$ initialization. The same filter $h^{(4)}$ is also used for both the rank-one experiments. Moreover, the uniform rank-one lattice has $N=M^2+1$ points and a generator
$(1,M)$, while the fixed-direction lattice has $N=M^2+1$ points and
generator $(1,S)$, with $S=126$. These choices satisfy the stability
condition for the second-order quantizers used here.

For display only, only the halftoned images are translated periodically by
$(\pi,\pi)$, so that the origin and the two torus seams appear at the
center. This translation changes neither the bit sequences nor the
reconstructions. 

In the rank-one panels, each bit is instead extended over its exact
fundamental rectangle and the cells are rasterized at resolution $4M\times
4M$, rather than assigned to the nearest Cartesian pixel. Figure \ref{fig:flower-halftones} shows that all the discretizations lead to good halftoned images; however, for the scheme $A_{4,4}$, one indeed observes line artifacts along the horizontal and vertical initialization boundaries. As expected, for the uniform rank-one scheme, this effect is no longer visible, but the general appearance is slightly coarser. For the fixed direction rank-one lattice with rectangular pixels, not only are the boundary effects removed, but also the shading of the grey areas looks more uniform.

\begin{figure}[t]
    \centering
    \begin{subfigure}[t]{0.315\textwidth}
        \centering
        \includegraphics[width=\linewidth]{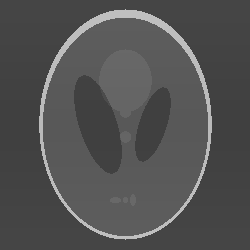}
        \caption{Normalized input.}
    \end{subfigure}\hfill
    \begin{subfigure}[t]{0.315\textwidth}
        \centering
        \includegraphics[width=\linewidth]{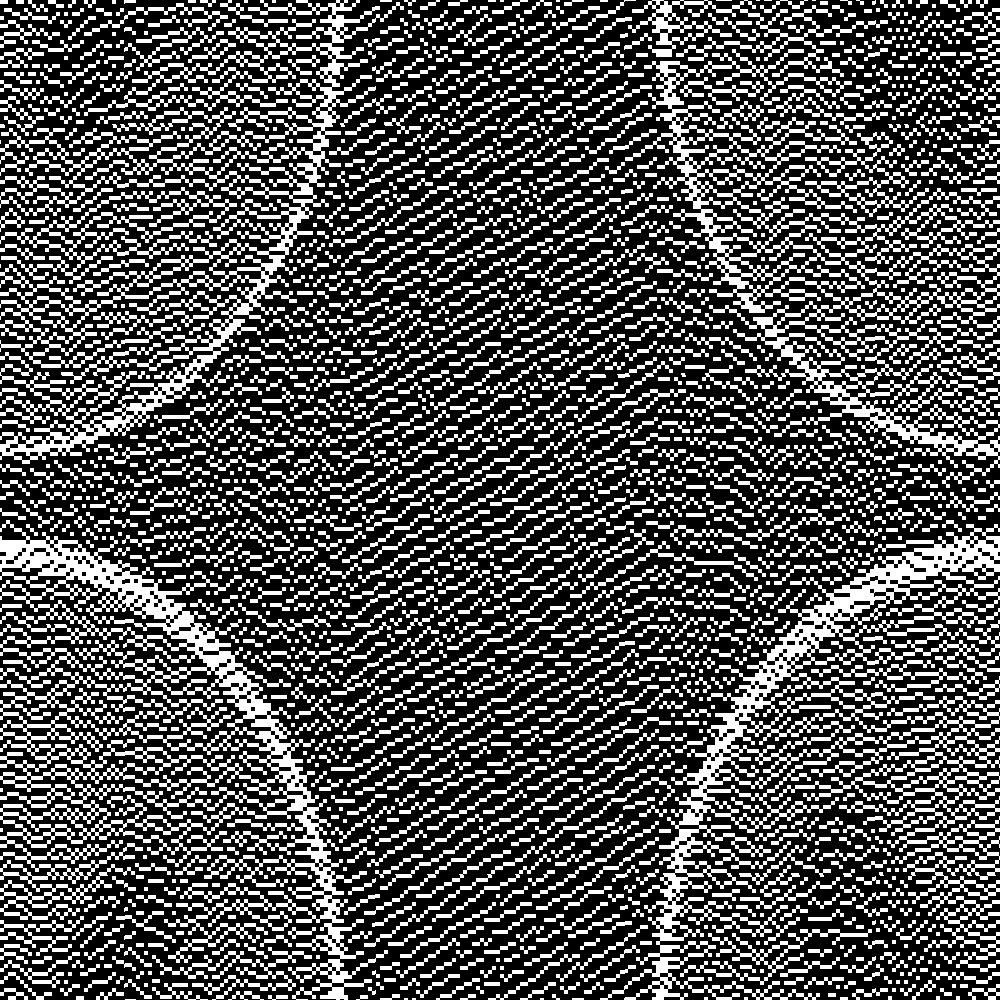}
        \caption{Uniform, before update.}
    \end{subfigure}\hfill
    \begin{subfigure}[t]{0.315\textwidth}
        \centering
        \includegraphics[width=\linewidth]{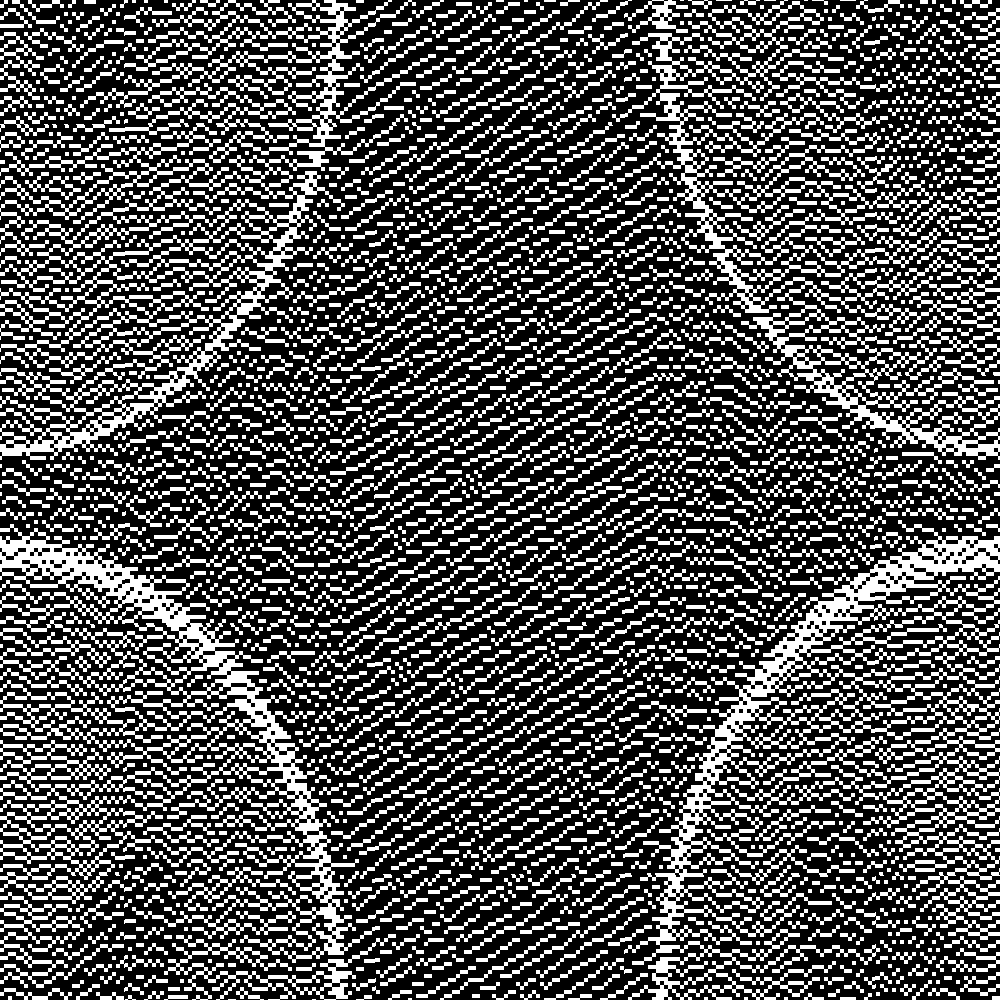}
        \caption{Uniform, after update.}
    \end{subfigure}

    \vspace{0.45em}

    \begin{subfigure}[t]{0.315\textwidth}
        \centering
        \includegraphics[width=\linewidth]{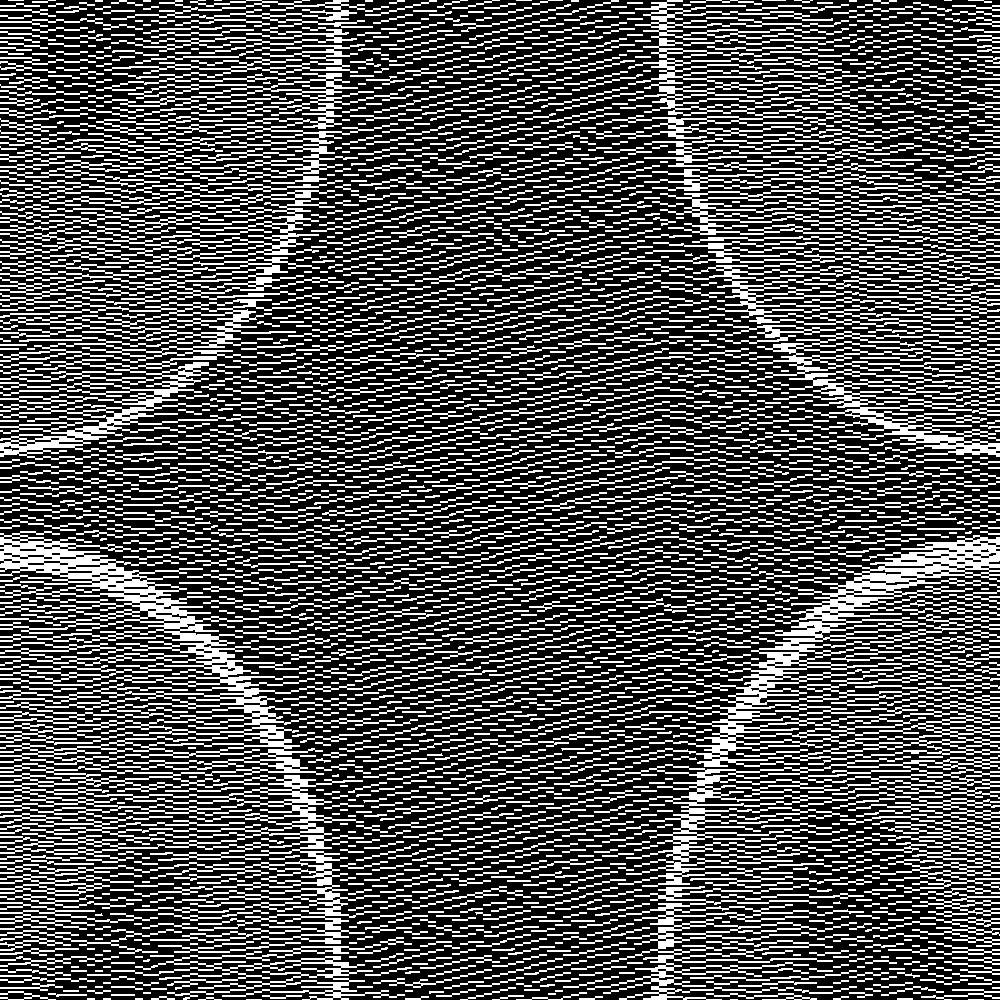}
        \caption{Fixed direction, before update.}
    \end{subfigure}\hfill
    \begin{subfigure}[t]{0.315\textwidth}
        \centering
        \includegraphics[width=\linewidth]{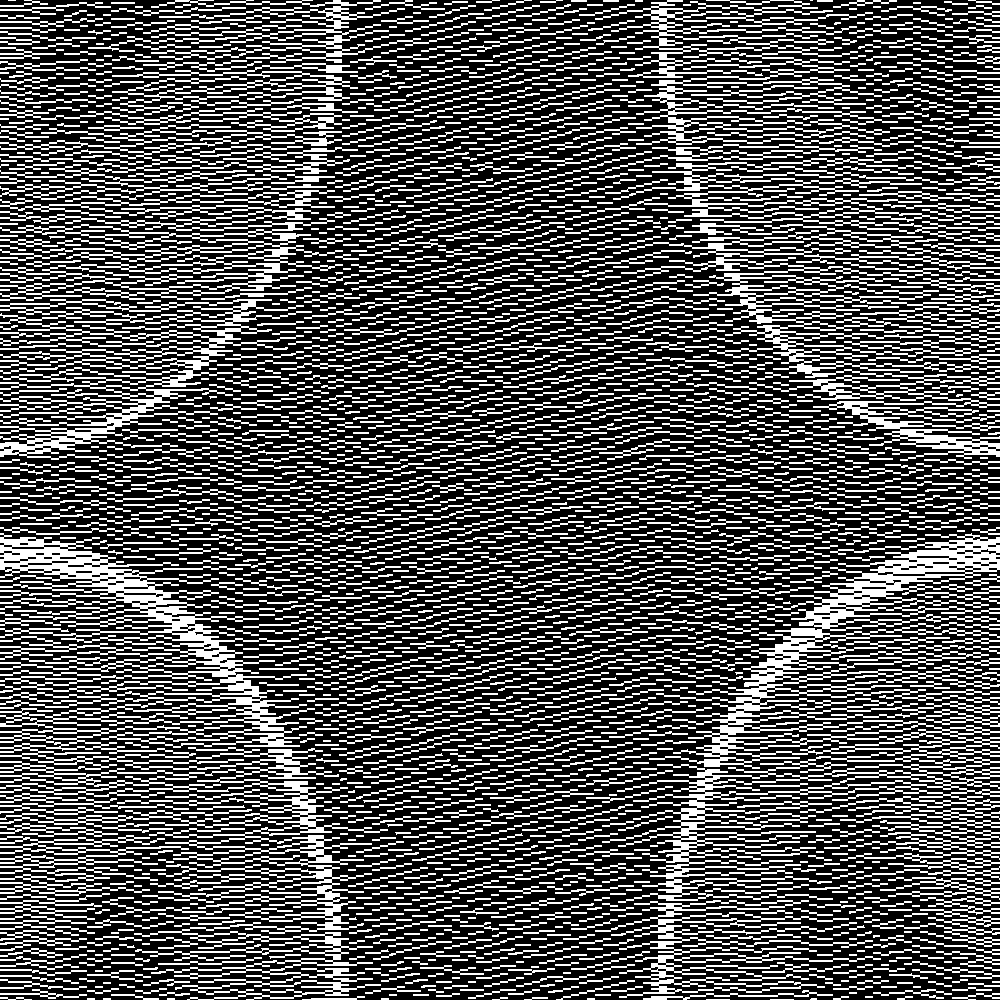}
        \caption{Fixed direction, after update.}
    \end{subfigure}\hfill
    \begin{subfigure}[t]{0.315\textwidth}
        \centering
        \includegraphics[width=\linewidth]{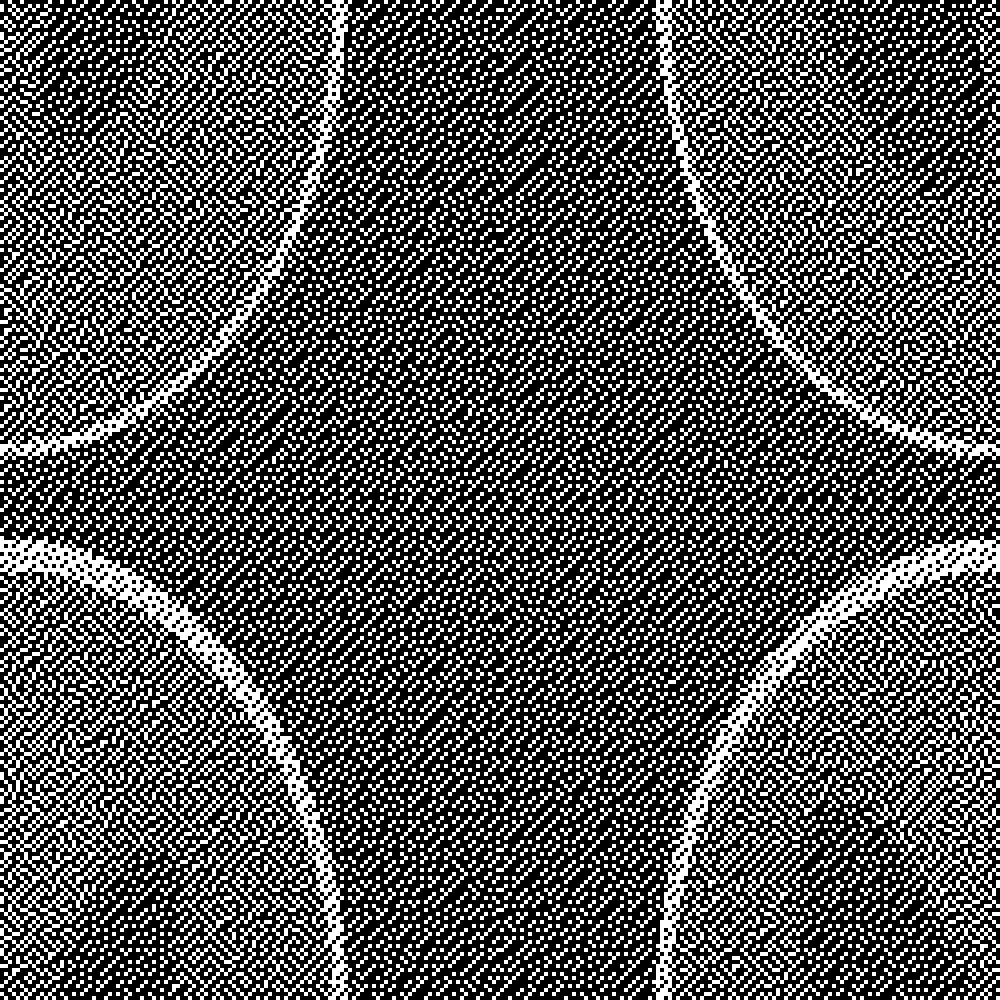}
        \caption{$A_{4,4}$.}
    \end{subfigure}
    
    \caption{The normalized periodic input, the four rank-one halftoned images,
    and the $A_{4,4}$ halftoned image. Both rank-one schemes use $h^{(4)}$.
    The rank-one bits are displayed on fundamental parallelogram cells, and
    the torus origin is shifted to the center.}
    \label{fig:flower-halftones}
\end{figure}

\section{Conclusions and outlook}
In this work, we studied the boundary mismatch produced by Sigma–Delta schemes on the
two-dimensional torus. By arranging the samples along a single closed rank-one lattice,
we replaced the multiple boundary contributions arising from separately initialized rows
and columns by a single terminal contribution, while preserving exact reconstruction of
bandlimited functions. For the uniform lattice, we obtained first- and second-order error
bounds of order $N^{-1/2}$ and $N^{-1}$
, respectively. A suitable global update eliminates the
terminal mismatch and reduces the spatial localization of the error. For fixed-direction lattices, compensating for the constant update leads to the improved rates $N^{-1}$ and $N^{-2}$, at the cost of a less uniform sampling geometry. The numerical experiments support the
reduction of boundary artifacts predicted by this construction.
A natural direction for future work is to extend this idea to other closed domains
and manifolds. The general objective is to identify suitable sampling curves or orderings
that allow standard one-dimensional Sigma–Delta schemes to be applied while remaining
compatible with the geometry and reconstruction theory of the underlying domain. Further
questions include the extension to higher-order Sigma–Delta schemes, the design of schemes
for the classical Cartesian sampling, more general stable feedback filters, and broader signal
models beyond bandlimited functions.

\bibliographystyle{unsrt}
\bibliography{references}

\end{document}